\documentclass[conference,letterpaper]{IEEEtran}

\usepackage{cite}
\usepackage[pdftex]{graphicx,color}
\usepackage[cmex10]{amsmath}

\newtheorem{theorem}{Theorem}
\newtheorem{remark}{Remark}
\newtheorem{corollary}{Corollary}

\begin{document}
\title{Outer Bounds for the Capacity Region of a Gaussian Two-way Relay Channel}

\author{\IEEEauthorblockN{Ishaque Ashar K, Prathyusha V, Srikrishna Bhashyam and Andrew Thangaraj}
\IEEEauthorblockA{Department of Electrical Engineering, Indian Institute of Technology Madras, India\\
Email: \{skrishna,andrew\}@ee.iitm.ac.in}
}
\maketitle

\begin{abstract}
  \let\thefootnote\relax\footnotetext{This work was supported in part
    by Renesas Mobile Corporation. This work was performed at the
    Department of Electrical Engineering, IIT Madras. Ishaque Ashar K
    is currently with Redpine Signals, Inc., Hyderabad, India.}  We
  consider a three-node half-duplex Gaussian relay network where two
  nodes (say $a$, $b$) want to communicate with each other and the
  third node acts as a relay for this two-way communication. Outer
  bounds and achievable rate regions for the possible rate pairs
  $(R_a, R_b)$ for two-way communication are investigated. The modes
  (transmit or receive) of the half-duplex nodes together specify the
  {\em state} of the network. A relaying protocol uses a specific
  sequence of states and a coding scheme for each state. In this
  paper, we first obtain an outer bound for the rate region of all
  achievable $(R_a, R_b)$ based on the half-duplex cut-set bound. This
  outer bound can be numerically computed by solving a linear
  program. It is proved that at any point on the boundary of the outer
  bound only four of the six states of the network are used. We then
  compare it with achievable rate regions of various known
  protocols. We consider two kinds of protocols: (1) protocols in
  which all messages transmitted in a state are decoded with the
  received signal in the same state, and (2) protocols where
  information received in one state can also be stored and used as
  side information to decode messages in future states. Various
  conclusions are drawn on the importance of using all states, use of
  side information, and the choice of processing at the relay.  Then,
  two analytical outer bounds (as opposed to an optimization problem
  formulation) are derived. Using an analytical outer bound, we
  obtain the symmetric capacity within 0.5 bits for some channel
  conditions where the direct link between nodes $a$ and $b$ is weak.

\end{abstract}

\IEEEpeerreviewmaketitle

\section{Introduction}
Two-way or bidirectional relaying has attracted significant interest
recently
\cite{PopYom07,BaiChu08,ZhaLie09,WilNar10,kim2008performance,kim2011achievable,tian2012asymmetric,GonYueWan11,icc12sixstate}. In
two-way relaying, a relay node assists in bidirectional communication
between two nodes.
\begin{figure}[!hb]
\centering
\resizebox{1.8in}{!}{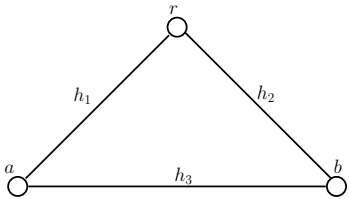}
\caption{Half-duplex Two-way Gaussian Relay Channel with Direct Link.}
\label{network}
\end{figure}
Fig. \ref{network} shows a three-node two-way relay network in which
nodes $a$ and $b$ want to communicate with each other at rates $R_a$ and $R_b$, respectively, while node $r$ acts as 
the relay. While two-way relaying without the direct link between
nodes $a$ and $b$ is studied in
\cite{PopYom07,BaiChu08,ZhaLie09,WilNar10}, the more general two-way
relaying with the direct link is studied in
\cite{kim2008performance,kim2011achievable,tian2012asymmetric,GonYueWan11,icc12sixstate}. In our
work, we consider two-way relaying with the direct link.


In \cite{kim2008performance}, achievable rate regions are derived for
three protocols, namely the multiple access broadcast (MABC) protocol,
time division broadcast (TDBC) protocol, and the hybrid broadcast
(HBC) protocol under the decode-and-forward (DF) relaying scheme. In
\cite{kim2011achievable}, the TDBC and MABC protocols have been
studied under other relaying schemes such as amplify-and-forward (AF),
compress-and-forward (CF), mixed forward (MF) and Lattice forward
(LF). A partial DF protocol, which is a superposition of DF and CF was
studied in \cite{schnurr2008achievable}. In \cite{tian2012asymmetric},
a three-phase cooperative MABC protocol (CoMABC), which
outperforms the MABC and TDBC protocols in terms of sum rate in
asymmetric channel conditions, was proposed.  A transmission scheme based on doubly
nested lattice codes was also proposed in \cite{tian2012asymmetric},
i.e., the CoMABC protocol uses the Lattice forward relaying strategy
and is not a DF relaying protocol.

The number of phases (or network states) in all the above protocols is
between 2 and 4. Each phase (or network state) refers to a particular
configuration of transmit and receive modes for the half-duplex
nodes.  In the absence of the direct link between $a$ and $b$, the 2
states (multiple access and broadcast) in the MABC protocol are
sufficient. In the presence of the direct link between $a$ and $b$,
more states are required. The TDBC, HBC, and CoMABC protocols use more
states to improve performance. However, all possible states have not
been considered in these protocols.  Since each node can be in transmit or receive mode,
the 3-node half-duplex two-way relay channel has $2^3$ ($= 8$)
possible states. Of these 8, the two states in which all nodes are
receiving or all nodes are transmitting are not useful in information
transfer. Therefore, there are 6 useful states as shown in
Fig. \ref{states}. A protocol that uses all these 6 states has been proposed in
\cite{GonYueWan11}, and the achievable rate region of this protocol has
been derived. The improvement in achievable rate region with respect
to the HBC protocol has been shown in \cite{GonYueWan11}.  The
6-state protocol has also been presented independently in
\cite{icc12sixstate} in the context of a degrees of freedom analysis
of two-way relaying with multiple-antenna nodes.
\begin{figure}[!ht] 
\centering
\resizebox{2.6in}{!}{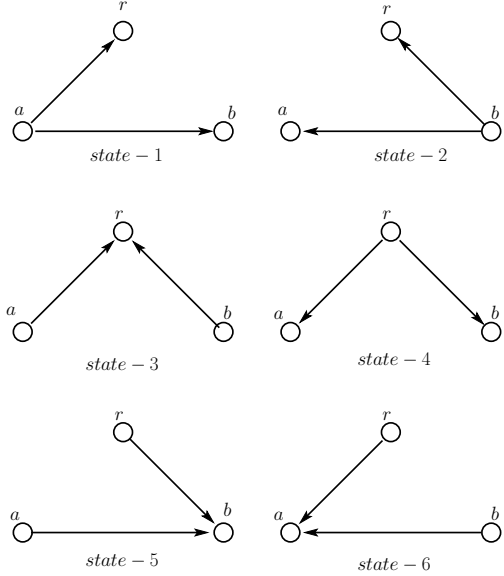}
\caption{Possible states in a Half-duplex Two-way Relay Channel}
\label{states}
\end{figure}

In this paper, we first obtain an outer bound for the two-dimensional
capacity region (of all possible rate pairs $(R_a,R_b)$) that is valid for all
relaying protocols.  This bound is derived based on the half-duplex
cutset bound in \cite{khojastepour2003capacity} and is more general
than the bounds in \cite{kim2008performance,kim2011achievable,tian2012asymmetric},which are for specific protocols --
TDBC, MABC, HBC, and CoMABC. The
outer bounds for TDBC, MABC, and CoMABC protocols consider at most 3
states and the bound for the HBC protocol considers 4 states. Further,
the outer bound for the HBC protocol in \cite{kim2008performance} is
difficult to compute and has not been computed. In contrast, the new outer bound can be numerically computed by solving a linear
program. It is also proved that at any point on the boundary of the
outer bound only four of the six states of the network are used. No outer bound is derived or
compared with for the 6-state protocol in \cite{GonYueWan11}. 

We then compare the new outer bound with achievable rate regions of the best known
protocols so far -- HBC, CoMABC, and 6-state protocols. We consider
two kinds of protocols: (1) protocols in which all messages
transmitted in a state are decoded with the received signal in the
same state, and (2) protocols where information received in one state
can also be stored and used as side information to decode messages in
future states. Through these comparisons, various conclusions are drawn on the importance of
using all states, use of side information, and the choice of
processing at the relay. In this context, a simple 6-state DF protocol  
without side information using all six states is also
presented and compared with. 

Finally, two analytical outer bounds (as
opposed to an optimization problem formulation) are derived. Using an
analytical outer bound, we obtain the symmetric ($R_a=R_b$) capacity within 0.5 bits for
some channel conditions where the direct link between nodes $a$ and
$b$ is weak.

\section{System Model}
Consider two nodes $a$ and $b$ communicating with each other in the network shown in Fig. \ref{network}. The relay node $r$ assists this communication by receiving the information from these nodes and forwarding to the desired destination. All nodes are half-duplex nodes with a receiver noise variance of $N$. For simplicity, each node is assumed to have the same transmit power $P$ for each state. Let $h_1$, $h_2$ and $h_3$ be the gains of  channels $a$-$r$, $b$-$r$ and $a$-$b$, respectively. The SNRs of these channels are denoted $\gamma_1=\frac{h_1^2P}{N}$, $\gamma_2=\frac{h_2^2P}{N}$ and $\gamma_3=\frac{h_3^2P}{N}$. We use $R_a$ and $R_b$ to denote the rate of data transmission (bits per channel use) from node $a$ to node $b$ and from node $b$ to node $a$, respectively. We consider $\gamma_3 \leq \gamma_1$, $\gamma_3\leq\gamma_2$, i.e., the direct link between $a$ and $b$ is weaker than the links $a$-$r$ and $b$-$r$. Further, let $\gamma_1 \leq \gamma_2$, without loss of generality.  Let $\mathcal{C}(\gamma)\stackrel{\Delta}{=}\log_2(1+\gamma)$ represent the capacity of a complex gaussian channel with SNR of $\gamma$.

\section{Outer bound for any protocol}
In this section, we derive an outer bound for the capacity region of
the half-duplex two-way Gaussian relay channel. This bound is an outer bound
for any relaying protocol for the two-way relay network irrespective of the number of states 
and the relaying scheme (eg. AF, DF, CF) used at the relay node. This outer bound is derived using
the half-duplex cut-set bound for information flow in an
arbitrary half-duplex relay network in
\cite{khojastepour2003capacity}. In this two-way relay network, we
have two flows, one from $a$ to $b$ and another from $b$ to $a$. We
derive the outer bound for the two flows considering two cuts that
separate nodes $a$ and $b$ and information flow in both directions
across the cuts. Combining these four bounds (2 cuts $\times$ 2
directions), we get the rate region outer bound.

Let $\lambda_i$, $i=1,2,\ldots,6$, denote the fraction of channel uses in network state $i$.  
We obtain the outer bound in the following manner. (1) For a given real number $k$, upper 
bound the maximum possible rate $R_a$ subject to $R_a = k R_b$, and (2) Vary $k$ and 
determine the whole region. For example, we can set $k = \tan{\theta}$ and vary $\theta$ from 0 to $90^o$.

\begin{theorem}
 Given $R_a = k R_b$ for some $k \ge 0$, the maximum possible $R_b$ is upperbounded by $C_{bk}$ obtained by solving the following linear program:
\[
C_{bk} = \max_{R_b, \{\lambda_i\}} R_b
\]
subject to
\begin{equation}
 \begin{split}
  kR_b &\leq \lambda_1\mathcal{C}(\gamma_1+\gamma_3)+\lambda_3\mathcal{C}(\gamma_1) +\lambda_5\mathcal{C}\left(\gamma_3\right),\\
  kR_b &\leq \lambda_1\mathcal{C}(\gamma_3)+\lambda_4\mathcal{C}(\gamma_2) +\lambda_5\mathcal{C}\left(\left(\sqrt{\gamma_2}+\sqrt{\gamma_3}\right)^2\right),\\
  R_b &\leq \lambda_2\mathcal{C}(\gamma_2+\gamma_3)+\lambda_3\mathcal{C}(\gamma_2) +\lambda_6\mathcal{C}\left(\gamma_3\right),\\
  R_b &\leq \lambda_2\mathcal{C}(\gamma_3)+\lambda_4\mathcal{C}(\gamma_1) +\lambda_6\mathcal{C}\left(\left(\sqrt{\gamma_1}+\sqrt{\gamma_3}\right)^2\right),\\
  &\sum_{i = 1}^{6} \lambda_i \le 1,\;\lambda_i \ge 0,\;R_b \ge 0.\\
 \end{split}
\label{constraints}
\end{equation}
\label{th1}
\end{theorem}
\begin{proof}
For a half-duplex relay network with $m$ states, for which the sequence of states is fixed with asymptotic fraction of time $\lambda_i$ in state $i$, any achievable rate $R$ of information flow  is upper bounded as follows \cite{khojastepour2003capacity}:
\begin{equation}
R \le   \min_S \sum_{i=1}^m \lambda_i I(X_S;Y_{S^c}|X_{S^c}, i),
\end{equation}
where a cut partitions the set of nodes into sets $S$ and $S^c$ such
that the source nodes are in $S$, the destination nodes are in $S^c$, and $S^c$
is the complement of $S$. We use this bound for $R_a$ and $R_b$ as
follows. In order to bound $R_a$ (rate of information flow from $a$ to
$b$), we consider the cuts defined by $S_1=\{a\}$, and $S_2=\{a,
r\}$. Similarly to bound $R_b$, we use $S_3=\{b\}$ and $S_4=\{b, r\}$.

Using $S_1$ and $S_2$ as defined above, we get the following bounds.
\begin{equation}
R_a \le  \min \{R_{a1}, R_{a2} \}, 
\label{bounda}
\end{equation}
where 
\begin{equation}
 \begin{split}
  R_{a1} &= \lambda_1 I\left(X_a;Y_r,Y_b|i = 1\right)+\lambda_3 I\left(X_a;Y_r|X_b,i = 3\right)\\
  &+\lambda_5 I\left(X_a;Y_b|X_r, i = 5\right),
 \end{split}
\end{equation}
\begin{equation}
 \begin{split}
  R_{a2} &= \lambda_1 I\left(X_a;Y_b| i = 1\right)+\lambda_4 I\left(X_r;Y_b| i = 4\right)\\
  &+\lambda_5 I\left(X_a,X_r;Y_b| i = 5\right).
 \end{split}
\end{equation}
Similarly, using $S_3$, and $S_4$, we get
\begin{equation}
R_b \le   \min \{R_{b3}, R_{b4} \}, 
\label{boundb}
\end{equation}
where 
\begin{equation}
 \begin{split}
  R_{b3} = \lambda_2 I\left(X_b;Y_r,Y_a| i = 2\right)&+\lambda_3 I\left(X_b;Y_r|X_a, i = 3\right)\\
  &+\lambda_6 I\left(X_b;Y_a|X_r, i = 6\right)
 \end{split}
\end{equation}
\begin{equation}
 \begin{split}
  R_{b4} = \lambda_2 I\left(X_b;Y_a| i = 2\right)&+\lambda_4 I\left(X_r;Y_a| i = 4\right)\\
  &+\lambda_6 I\left(X_b,X_r;Y_a| i = 6\right)
 \end{split}
\end{equation}
Now, we can upper bound the various mutual information terms as in \cite{khojastepour2003capacity} to get:
\begin{equation}
 \begin{split}
  &I\left(X_a;Y_r,Y_b|i = 1\right) \leq \mathcal{C}(\gamma_1+\gamma_3),\\
  &I\left(X_b;Y_r,Y_a|i = 2\right) \leq \mathcal{C}(\gamma_2+\gamma_3),\\
  &I\left(X_a;Y_b|X_r, i = 5\right) \leq \mathcal{C}\left((1-\rho_5^2)\gamma_3\right),\\
  &I\left(X_b;Y_a|X_r, i = 6\right) \leq \mathcal{C}\left((1-\rho_6^2)\gamma_3\right),
 \end{split}
\label{rho56-0}
\end{equation}
where $\rho_5$ is the correlation coefficient between $X_a$ and $X_r$ and $\rho_6$ is the correlation coefficient between $X_b$ and $X_r$.
\begin{equation}
  I\left(X_a;Y_r|X_b, i = 3\right) \leq \mathcal{C}\left(\gamma_1\right), I\left(X_b;Y_r|X_a, i = 3\right) \leq \mathcal{C}\left(\gamma_2\right)
\end{equation}
Similarily, we get
\begin{equation}
 \begin{split}
  I\left(X_a;Y_b| i = 1\right) &\leq \mathcal{C}\left(\gamma_3\right),  I\left(X_r;Y_b| i = 4\right) \leq \mathcal{C}\left(\gamma_2\right),\\
  I\left(X_b;Y_a| i = 2\right) &\leq \mathcal{C}\left(\gamma_3\right),  I\left(X_r;Y_a| i = 4\right) \leq \mathcal{C}\left(\gamma_1\right),\\ 
  I\left(X_a,X_r;Y_b| i = 5\right) &\leq \mathcal{C}\left(\gamma_2+\gamma_3+2\rho_5\sqrt{\gamma_2\gamma_3}\right),\\ 
  I\left(X_b,X_r;Y_a| i = 6\right) &\leq \mathcal{C}\left(\gamma_1+\gamma_3+2\rho_6\sqrt{\gamma_1\gamma_3}\right).\\ 
 \end{split}
\label{rho56-1}
\end{equation}
Using these mutual information upper bounds in the bounds (\ref{bounda}) and (\ref{boundb}), we get: 
\begin{equation}
 \begin{split}
  R_a &\leq \lambda_1\mathcal{C}(\gamma_1+\gamma_3)+\lambda_3\mathcal{C}(\gamma_1)+\lambda_5\mathcal{C}\left(\left(1-\rho_5^2\right)\gamma_3\right),\\
  R_a &\leq \lambda_1\mathcal{C}(\gamma_3)+\lambda_4\mathcal{C}(\gamma_2)+\lambda_5\mathcal{C}\left(\gamma_2+\gamma_3+2\rho_5\sqrt{\gamma_2\gamma_3}\right),\\
  R_b &\leq \lambda_2\mathcal{C}(\gamma_2+\gamma_3)+\lambda_3\mathcal{C}(\gamma_2)+\lambda_6\mathcal{C}\left(\left(1-\rho_6^2\right)\gamma_3\right),\\
  R_b &\leq \lambda_2\mathcal{C}(\gamma_3)+\lambda_4\mathcal{C}(\gamma_1)+\lambda_6\mathcal{C}\left(\gamma_1+\gamma_3+2\rho_6\sqrt{\gamma_1\gamma_3}\right),\\
 \end{split}
\end{equation}
Using $\mathcal{C}\left((1-\rho_5^2)\gamma_3\right) \le \mathcal{C}\left(\gamma_3\right)$ and $\mathcal{C}\left((1-\rho_6^2)\gamma_3\right) \le \mathcal{C}\left(\gamma_3\right)$ in (\ref{rho56-0}), $\rho_5, \rho_6 \le 1$ in (\ref{rho56-1}), and $R_a = k R_b$, we get the constraints in (\ref{constraints}).
Since these constraints have to be satisfied by any achievable rate for a given set of $\lambda_i$'s, the upper bound $C_{ak}$ can be determined by maximizing $R_a$ over the $\lambda_i$'s subject to the constraints of the linear program in (\ref{constraints}). 
\end{proof}

The complete outer bound for the capacity region can be obtained by solving the above linear program for different values of the ratio $k$. It is worth noting here that this outer bound is valid for any relaying scheme (since it is based only on the cut-set bound and mutual information bounds) for the two-way relay channel and is not restricted to DF or other protocols presented here. 

Using results from linear programming, we can now argue that for any point on the boundary of the outer bound (i.e., for any $k$), at most four of the six $\lambda_i$'s are non-zero, i.e., at most 4 states are used. However, it should be noted that which states are used depends on $k$.
\begin{theorem}
For any $k$, there are at most four positive $\lambda_i$'s in the solution to the linear program in (\ref{constraints}). 
\end{theorem}
\begin{proof}
It is optimal to choose $\sum_{i = 1}^{6} \lambda_i = 1$, i.e., satisfy the last constraint with equality. There are 4 other inequality constraints. Adding 4 slack variables, we can convert these constraints to equality constraints. Now, there are $7+4 = 11$ variables in the linear program with 5 equality constraints. Therefore, from \cite{BerTsi97}, we know that there are at most 5 positive variables in the solution. Since $R_b$ should be positive, we can have at most 4 of the $\lambda_i$'s to be positive. 
\end{proof}

\vspace*{-1mm}
\subsection{Alternative Outer Bound}
An alternative outer bound can be obtained as follows: (1) For a given
real number $k$, upper bound the maximum possible weighted sum rate
$R_a + kR_b$, and (2) Vary $k$ and determine the whole region.

\begin{theorem}
For some $k \ge 0$, the maximum possible $R_a + kR_b$ is upperbounded by $C_{k}$ obtained by solving the following linear program:
\[
C_{k} = \max_{R_a, R_b, \{\lambda_i\}} R_a + kR_b, \mbox{~~subject to~~}
\]
{\allowdisplaybreaks\begin{equation}
 \begin{split}
  R_a &\leq \lambda_1\mathcal{C}(\gamma_1+\gamma_3)+\lambda_3\mathcal{C}(\gamma_1) +\lambda_5\mathcal{C}\left(\gamma_3\right),\\
  R_a &\leq \lambda_1\mathcal{C}(\gamma_3)+\lambda_4\mathcal{C}(\gamma_2) +\lambda_5\mathcal{C}\left(\left(\sqrt{\gamma_2}+\sqrt{\gamma_3}\right)^2\right),\\
  R_b &\leq \lambda_2\mathcal{C}(\gamma_2+\gamma_3)+\lambda_3\mathcal{C}(\gamma_2) +\lambda_6\mathcal{C}\left(\gamma_3\right),\\
  R_b &\leq \lambda_2\mathcal{C}(\gamma_3)+\lambda_4\mathcal{C}(\gamma_1) +\lambda_6\mathcal{C}\left(\left(\sqrt{\gamma_1}+\sqrt{\gamma_3}\right)^2\right),\\
  &\sum_{i = 1}^{6} \lambda_i \le 1,\;\lambda_i \ge 0,\;R_b \ge 0,; R_a \ge 0.\\
 \end{split}
\label{constraints2}
\end{equation}}
\label{th3}
\end{theorem}
For each $k$, we get a line. As $k$ is varied, we obtain a set of
lines that define an outer bound to the rate region. This
outer bound is equivalent to the one obtained in Theorem 1.

\section{Comparisons}
Now, we compare the outer bound with the achievable rate regions of various protocols. We consider the following two classes of protocols: (1) Simple DF protocols with no side information across states, (2) protocols with side information across states. In addition to existing protocols, we also introduce a 6-state DF protocol with no side information across states and compare its achievable rate region with other protocols. This provides further insight on the importance of using all states and on the use of side information. We briefly describe the achievable rate regions here and then provide numerical results. 

\subsection{DF protocols with no side information}
\label{simpleDF}

\subsubsection{MABC protocol}
The 3-node half-duplex relay network has six useful states as shown in
Fig. \ref{states}. The MABC protocol uses states
3 and 4. The achievable rate region for MABC, as been obtained in
\cite{kim2008performance}, is described by the following inequalities:
\begin{equation}
 \begin{split}
  R_a &\leq \min\{\lambda_3\mathcal{C}(\gamma_1),\lambda_4\mathcal{C}(\gamma_2)\},\\
  R_b &\leq \min\{\lambda_3\mathcal{C}(\gamma_2),\lambda_4\mathcal{C}(\gamma_1)\},\\
  R_a&+R_b \leq \lambda_3\mathcal{C}(\gamma_1+\gamma_2).
 \end{split}
\end{equation}

\subsubsection{6-state DF Protocol}
The proposed 6-state DF protocol uses all the 6 useful states. 
To optimize the fractions $\lambda_i$, $i=1,2,\ldots,6$, we proceed as follows. 
Let $Z_{kl}^i$ denote the rate of information flow (bits per channel use) from node
$k$ to node $l$ in network state $i$. The flow rates $Z_{kl}^i$ are constrained by the topology of the network state and the coding employed.  States 3, 5 and 6
form Multiple Access (MAC) channels, while states 1 and 2 are Broadcast Channels (BCs). For these states, we employ standard capacity-achieving coding methods \cite{cover:book91}. In state 4,
the relay node $r$ broadcasts the information it received from nodes $a$ and $b$
in previous states. Thus, in state 4, each
receiver already knows the message that it transmitted to the relay,
i.e., state 4 is a BC with two receivers knowing the message intended for the other user {\it a priori}. The capacity region for state 4 was determined in \cite{wu2007broadcasting}, and we employ the coding scheme given there.

The achievable rate region of the half-duplex bidirectional relay channel with the 6-state DF protocol is the closure of the set of all points $(R_a,R_b)$ satisfying the following constraints:
\begin{equation*}
 \begin{split}
  R_a &=Z_{ar}^1+Z_{ab}^1+Z_{ab}^5+Z_{ar}^3,\\
  R_b &=Z_{br}^2+Z_{ba}^2+Z_{ba}^6+Z_{br}^3,\\
  \sum_{i=1}^6 & \lambda_i = 1, 0\leq \lambda_i \leq 1, 0 \leq \alpha_1,\alpha_2 \leq 1.\\
 \end{split}
\end{equation*}

{\em State 1 rate constraints:}
\begin{equation*}
   Z_{ar}^1 \leq \lambda_1\mathcal{C}(\alpha_1\gamma_1), Z_{ab}^1 \leq \lambda_1\mathcal{C}\left(\frac{(1-\alpha_1)\gamma_3}{1+\alpha_1\gamma_3}\right).
\end{equation*}

{\em State 2 rate constraints:}
\begin{equation*}
   Z_{br}^2 \leq \lambda_2\mathcal{C}(\alpha_2\gamma_2),   Z_{ba}^2 \leq \lambda_2\mathcal{C}\left(\frac{(1-\alpha_2)\gamma_3}{1+\alpha_2\gamma_3}\right).
\end{equation*}

{\em State 3 rate constraints:}
\begin{equation*}
   Z_{ar}^3 \leq \lambda_3\mathcal{C}(\gamma_1),\;\;   Z_{br}^3 \leq \lambda_3\mathcal{C}(\gamma_2), Z_{ar}^3+Z_{br}^3 \leq \lambda_3\mathcal{C}\left(\gamma_1+\gamma_2\right).
\end{equation*}

{\em State 4 rate constraints:}
\begin{equation*}
   Z_{ra}^4 \leq \lambda_4\mathcal{C}(\gamma_1), \;\;  Z_{rb}^4 \leq \lambda_4\mathcal{C}(\gamma_2).\\
\end{equation*}

{\em State 5 rate constraints:}
\begin{equation*}
   Z_{rb}^5 \leq \lambda_5\mathcal{C}(\gamma_2), \;\;  Z_{ab}^5 \leq \lambda_5\mathcal{C}(\gamma_3), Z_{rb}^5+Z_{ab}^5 \leq \lambda_5\mathcal{C}\left(\gamma_2+\gamma_3\right).
\end{equation*}

{\em State 6 rate constraints:}
\begin{equation*}
   Z_{ra}^6 \leq \lambda_6\mathcal{C}(\gamma_1),\;\;  Z_{ba}^6 \leq \lambda_6\mathcal{C}(\gamma_3), Z_{ra}^6+Z_{ba}^6 \leq \lambda_6\mathcal{C}\left(\gamma_1+\gamma_3\right).
\end{equation*}

{\em Flow Constraints:}
Information received at node $r$ from node $a$ should be equal to the information forwarded from node $r$ to node $b$. Similarly in the other direction also. Thus, we get the following equality constraints:
\begin{equation*}
  Z_{ar}^1+Z_{ar}^3 = Z_{rb}^5+Z_{rb}^4,~~~  Z_{br}^2+Z_{br}^3 = Z_{ra}^6+Z_{ra}^4.
\end{equation*}
In these constraints, $\alpha_1$ is the fraction of power used for the message from $a$ to $r$ in state 1, and $\alpha_2$ is the fraction of power used for the message from $a$ to $r$ in state 2. Numerical evaluation of this achievable rate region of the 6-state DF protocol is done for some illustrative examples in the numerical results section to show the importance of the various states in achieving different parts of the rate region.

\subsection{Protocols with side information across states}
\subsubsection{TDBC and HBC protocols}
TDBC \cite{kim2008performance} is a three phase protocol in which
states 1, 2 and 4 are used. In first phase (state 1), $a$ transmits at
a rate equal to the capacity of link between $a$ and $r$. At this time
node $b$ listens to this transmission and uses this information as
{\em side information} for decoding after phase 3 (state 4). Phase 2 is
similar to phase 1 in which $b$ transmits at a rate equal to the
capacity of link between $b$ and $r$ and $a$ listens to this. At the
end of each of these phases, relay node $r$ decode the messages and does 
a binning operation on these messages. In phase 3, $r$ transmits this
binned information to $a$ and $b$. Since $a$ and $b$ know the message
meant for the other destination, relay node $r$ can use XOR of these
messages for broadcasting. 

HBC is a four phase protocol in which states 1, 2, 3 and 4 are
used. States 1, 2 and 4 are used in the same way as in TDBC. State 3,
where terminals $a$ and $b$ transmit simultaneously (MAC) to relay $r$,
is also added. These messages are decoded at the relay and forwarded
as such in state 4. In HBC, state 4 is used for forwarding binned
messages from state 1 and 2 as well as for forwarding messages
received in state 3. The HBC protocol is always better than TDBC and
MABC protocols since they are special cases of the HBC protocol.  The
use of {\em side information} from one phase in decoding during
another phase provides improvement in TDBC and HBC over MABC for some
channel conditions.

The achievable rate region for HBC protocol has been obtained in
\cite{kim2008performance}. It is the closure of the set of all points
$(R_a,R_b)$ satisfying following constraints.
\begin{equation}
 \begin{split}
  R_a &\leq \min\{(\lambda_1+\lambda_3)\mathcal{C}(\gamma_1), \lambda_1\mathcal{C}(\gamma_3)+\lambda_4\mathcal{C}(\gamma_2)\},\\
  R_b &\leq \min\{(\lambda_2+\lambda_3)\mathcal{C}(\gamma_2), \lambda_2\mathcal{C}(\gamma_3)+\lambda_4\mathcal{C}(\gamma_1)\},\\
  R_a&+R_b \leq \lambda_1\mathcal{C}(\gamma_1) + \lambda_2\mathcal{C}(\gamma_2) + \lambda_3\mathcal{C}(\gamma_1+\gamma_2),\\ 
  &\sum_{i=1}^4 \lambda_i = 1, 0\le \lambda_i\le 1.
 \end{split}
\end{equation}

\subsubsection{6-state protocol}
States 1, 2, 3 and 4 are used in the HBC protocol. A 6-state
protocol that uses all 6 states has been proposed in
\cite{GonYueWan11}. This protocol is similar to the HBC protocol except that MAC states 5 and 6 are used before the use of state 4. Here, we arrive at the same achievable
region as that of the 6-state protocol in a slightly different way
than in \cite{GonYueWan11}. States 1, 2, 3 and 4 are used as in the
HBC protocol. In states 5 and 6, we have MAC transmissions from
\{$a$,$r$\} and \{$b$,$r$\} to node $b$ and $a$, respectively, i.e., in
states 5 and 6, there is a direct transmission between the terminal
nodes $a$ and $b$, and forwarding of messages received at relay in
previous phases. In the HBC protocol, the direct link between $a$ and $b$ is only
used to obtain side information to decode the transmission from the
relay. In the 6-state protocol, the direct link is also used to send messages.

The achievable rate region of the 6-state protocol is now the
closure of the set of all points $(R_a,R_b)$ satisfying the following
constraints:
\begin{equation*}
 \begin{split}
  R_a \leq Z_{ab}^5+ \min\{&(\lambda_1+\lambda_3)\mathcal{C}(\gamma_1),\\
  &\lambda_1\mathcal{C}(\gamma_3)+\lambda_4\mathcal{C}(\gamma_2)+Z_{rb}^5\},\\
  R_b \leq Z_{ba}^6+ \min\{&(\lambda_2+\lambda_3)\mathcal{C}(\gamma_2),\\
  &\lambda_2\mathcal{C}(\gamma_3)+\lambda_4\mathcal{C}(\gamma_1)+Z_{ra}^6\},\\
  R_a+R_b \leq \lambda_1\mathcal{C}(\gamma_1)+Z_{ab}^5+&\lambda_2\mathcal{C}(\gamma_2)+Z_{ba}^6+\lambda_3\mathcal{C}(\gamma_1+\gamma_2),\\
  &\sum_{i=1}^6 \lambda_i = 1, 0\le \lambda_i\le 1.
 \end{split}
\end{equation*}

\begin{equation*}
\begin{split}
  Z_{ab}^5 &\leq \lambda_5\mathcal{C}(\gamma_3), \;\; Z_{rb}^5 \leq \lambda_5\mathcal{C}(\gamma_2),  Z_{ab}^5+Z_{rb}^5 \leq \lambda_5\mathcal{C}(\gamma_2+\gamma_3).\\
  Z_{ba}^6 &\leq \lambda_6\mathcal{C}(\gamma_3),\;\;  Z_{ra}^6 \leq \lambda_6\mathcal{C}(\gamma_1), Z_{ba}^6+Z_{ra}^6 \leq \lambda_6\mathcal{C}(\gamma_1+\gamma_3).\\
\end{split}
\end{equation*}
The above achievable rate region is obtained by modifying the achievable
rate region for HBC protocol to include the effect of states 5 and 6
as well. The main points are summarized here. (1) $Z_{ab}^5$
represents the direct information transmission from node $a$ to node
$b$ in state 5. Similarly $Z_{ba}^6$ represents the direct information
transmission from node $b$ to node $a$ in state 6. (2) The first term
inside min\{\} in the constraints for $R_a$ and $R_b$ correspond to
the flow from source to relay and the second term in the constraints
for $R_a$ and $R_b$ correspond to the flow from relay to
destination. For source to relay flow, this protocol also uses the
same states as HBC protocol. For forwarding from relay, this protocol
uses states 5 and 6 also. Thus, while the first terms inside min\{\}
remain the same as HBC, the second term has an additional term
corresponding to flow from relay in states 5 and 6 for rates $R_a$ and
$R_b$, respectively. (3) The bound on $R_a+R_b$ corresponds to the sum
of all information flow from both sources. Thus, in state 3, it is
bounded by the MAC capacity bound, and in other states it is bounded by
individual links. In state 4, there is no information flow from
source.

Now, we can observe that the choice of $Z_{ab}^5$, $Z_{rb}^5$, $Z_{ba}^6$, and $Z_{ra}^6$ that maximize the achievable rate region are: $Z_{ab}^5 + Z_{rb}^5 =  \lambda_2\mathcal{C}(\gamma_2+\gamma_3)$, $Z_{ba}^6+Z_{ra}^6 = \lambda_6\mathcal{C}(\gamma_1+\gamma_3)$, $  Z_{ab}^5 = \lambda_5\mathcal{C}(\gamma_3)$, and $Z_{ba}^6 = \lambda_6\mathcal{C}(\gamma_3)$. Substituting this choice, we get the same region as in \cite{GonYueWan11}.

\subsubsection{CoMABC protocol}
The achievable rate
region for CoMABC protocol is taken from \cite{tian2012asymmetric}. We
are assuming $\gamma_3\leq\gamma_1\leq\gamma_2$. Thus, states 3, 4 and
6 are used in this protocol. The constraints for the CoMABC
protocol are given by
\begin{equation}
 \begin{split}
  &R_a \leq \min\{\lambda_3 R_{ar}^*,\lambda_4 \mathcal{C}(\gamma_2)\},\\
  &R_b \leq \min\{\lambda_3 R_{br}^*+\lambda_6\mathcal{C}(\gamma_3),\lambda_4\mathcal{C}(\gamma_1)+\lambda_6\mathcal{C}(\gamma_1+\gamma_3)\},\\
  &R_{ar}^*=\left[\log\left(\frac{\gamma_1}{\gamma_1+\gamma_2}+\gamma_1\right)\right]^+, \\
&R_{br}^*=\left[\log\left(\frac{\gamma_2}{\gamma_1+\gamma_2}+\gamma_2\right)\right]^+.\\
 \end{split}
\end{equation}

\subsection{Numerical Results}
\subsubsection{Comparison with DF protocols}
\begin{figure}[htb]
\centering
\includegraphics[width=2.4in]{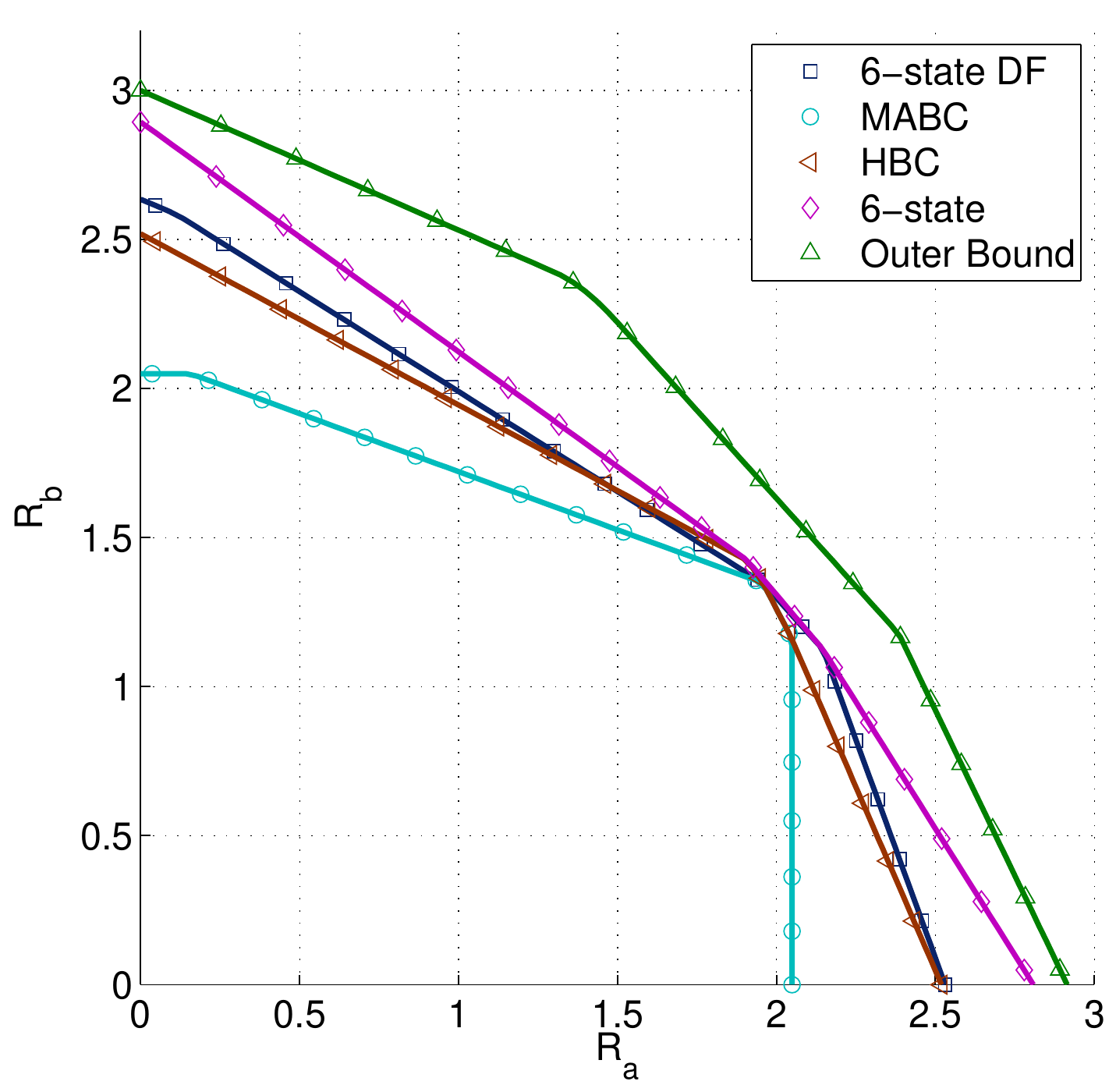}
\caption{Achievable rate region and Outer bound for various protocols:
  $\gamma_1 = 10$ dB, $\gamma_2 = 15$ dB, $\gamma_3 = 3$ dB}
\label{figure1}
\end{figure}
Figs. \ref{figure1} and \ref{figure2} show the comparison of different
DF protocols with the outer bound. The 6-state DF and HBC protocols are always
better than MABC. The 6-state DF protocol also achieves several rate pairs
that the HBC protocol cannot achieve even though the HBC protocol uses
side information across states. The HBC protocol also achieves some
rate pairs that the 6-state DF protocol cannot achieve (see Fig. \ref{figure2}
for small $R_b$). The 6-state protocol achieves a larger rate region
than all other protocols as expected since it uses all states as well
as the side information used in HBC protocol. The achievable rate
region of the 6-state protocol is closer to the outer bound in
Fig. \ref{figure1} where the SNRs are higher than in
Fig. \ref{figure2}.
\begin{figure}[ht]
\centering
\includegraphics[width=2.4in]{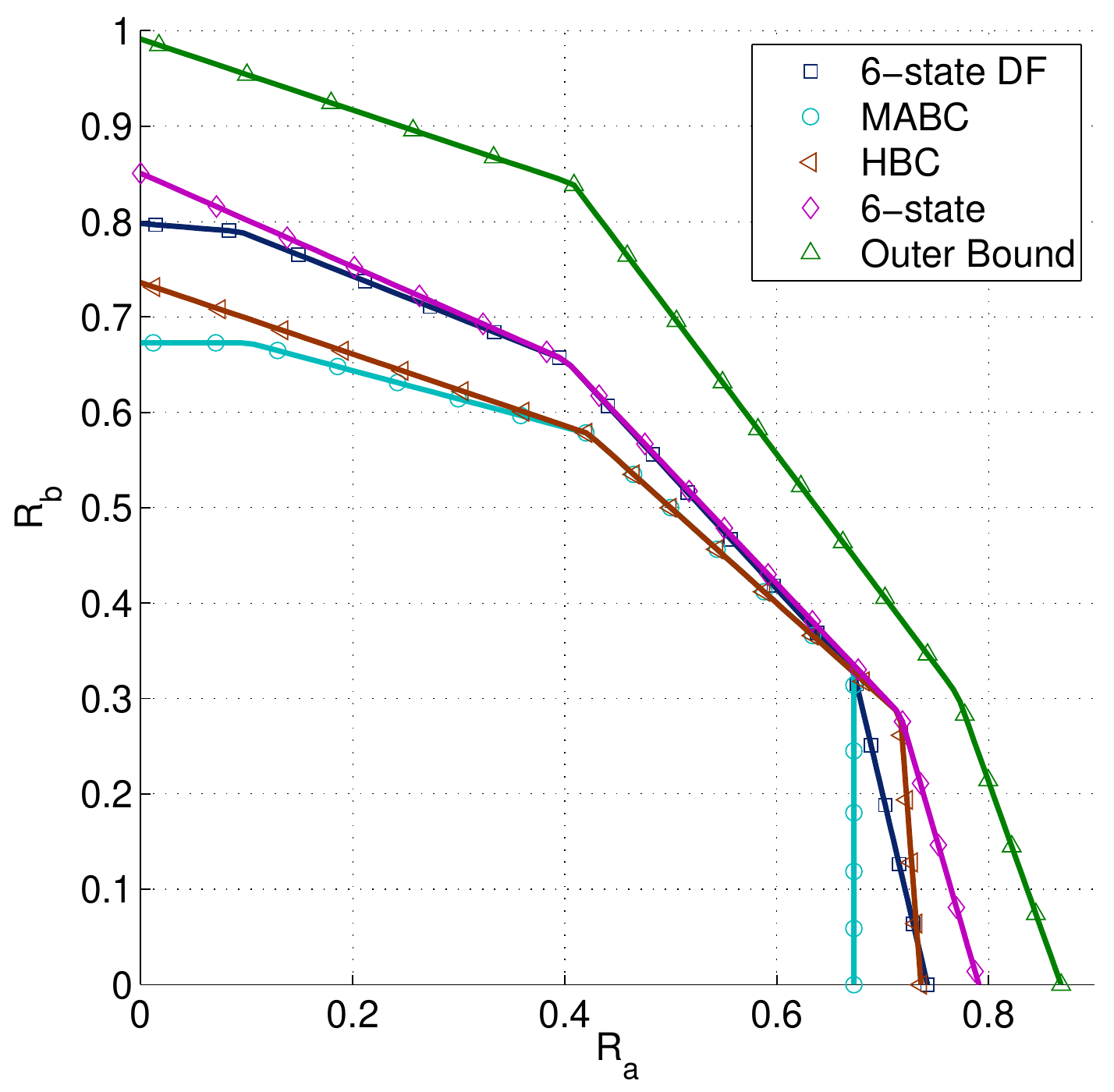}
\caption{Achievable rate region and Outer bound for various protocols:
  $\gamma_1 = 0$ dB, $\gamma_2 = 5$ dB, $\gamma_3 = -7$ dB}
\label{figure2}
\end{figure}

\subsubsection{Comparison with CoMABC protocol}
Figs. \ref{figure3}, \ref{figure5} and \ref{figure4} show the
comparison of the 6-state DF and 6-state protocols with the CoMABC
protocol. As mentioned earlier, CoMABC is not a DF protocol and the
relay forwards a estimated function of the two messages. Thus, it has
some advantage over DF protocols. However, the CoMABC protocol does
not use all states. Therefore, it does not achieve some rate pairs
that 6-state or 6-state DF can achieve. The CoMABC protocol is optimized to
maximize sum rate. Therefore, it performs well near the maximum sum
rate points. States 5 and 6 used in the 6-state DF and 6-state protocol are
very useful for achieving asymmetric rates. 
Fig. \ref{figure5} considers a high SNR scenario where $\gamma_1 = 30$
dB, $\gamma_2 = 35$ dB, and $\gamma_3 = 13$ dB. In this case, the
CoMABC protocol almost achieves the outer bound for symmetric rates,
while the 6-state protocol is close to the outer bound for asymmetric
rates. Fig. \ref{figure4} considers a low SNR scenario where $\gamma_1 = 0$
dB, $\gamma_2 = 5$ dB, and $\gamma_3 = -7$ dB. In this case, the
CoMABC protocol is worse than the 6-state protocol except near $R_a=0$. 
\begin{figure}[htb]
\centering
\includegraphics[width=2.4in]{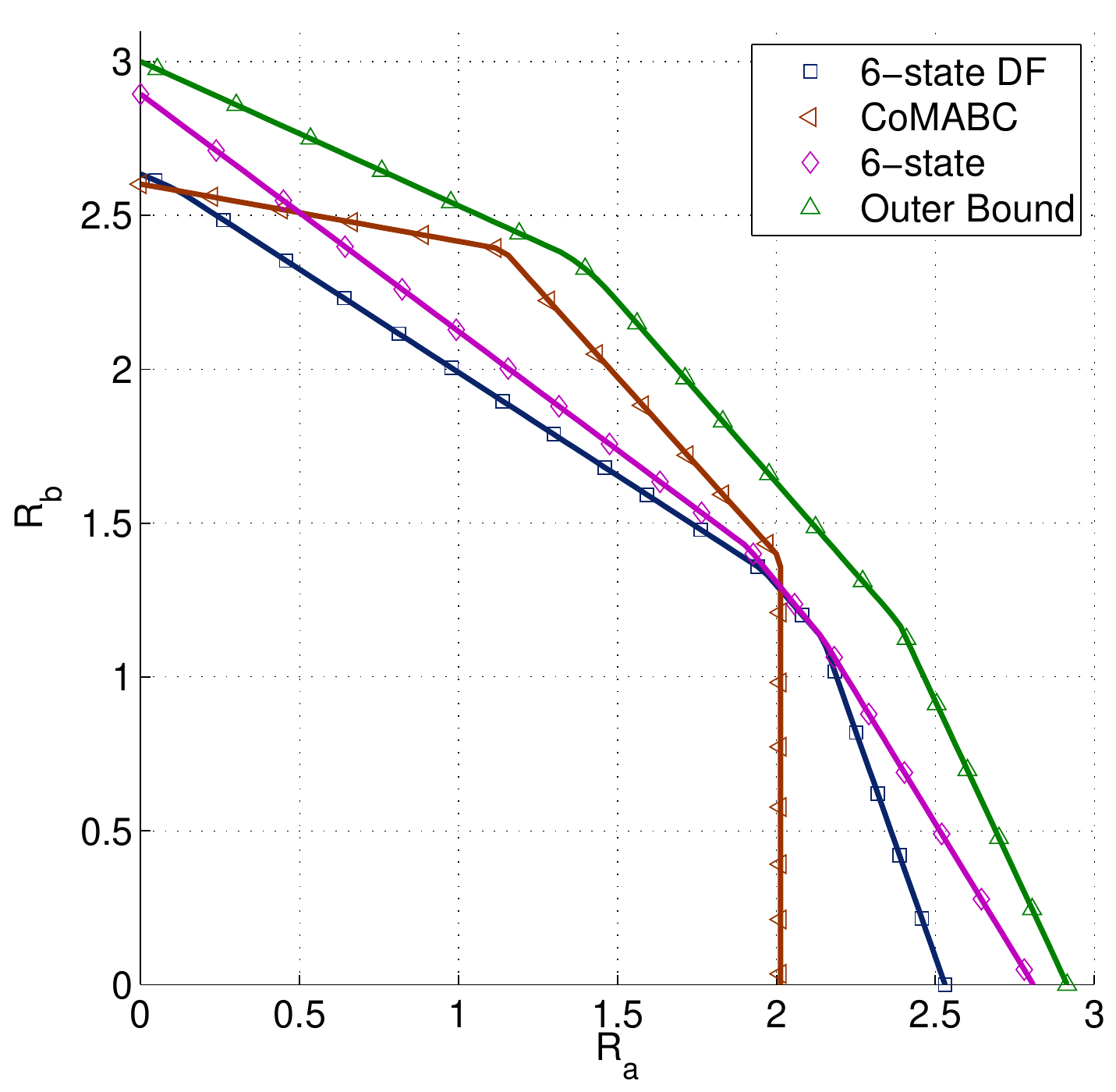}
\caption{Comparison with CoMABC: $\gamma_1 = 10$ dB, $\gamma_2 = 15$ dB, $\gamma_3 = 3$ dB}
\label{figure3}
\end{figure}
\begin{figure}[htb]
\centering
\includegraphics[width=2.4in]{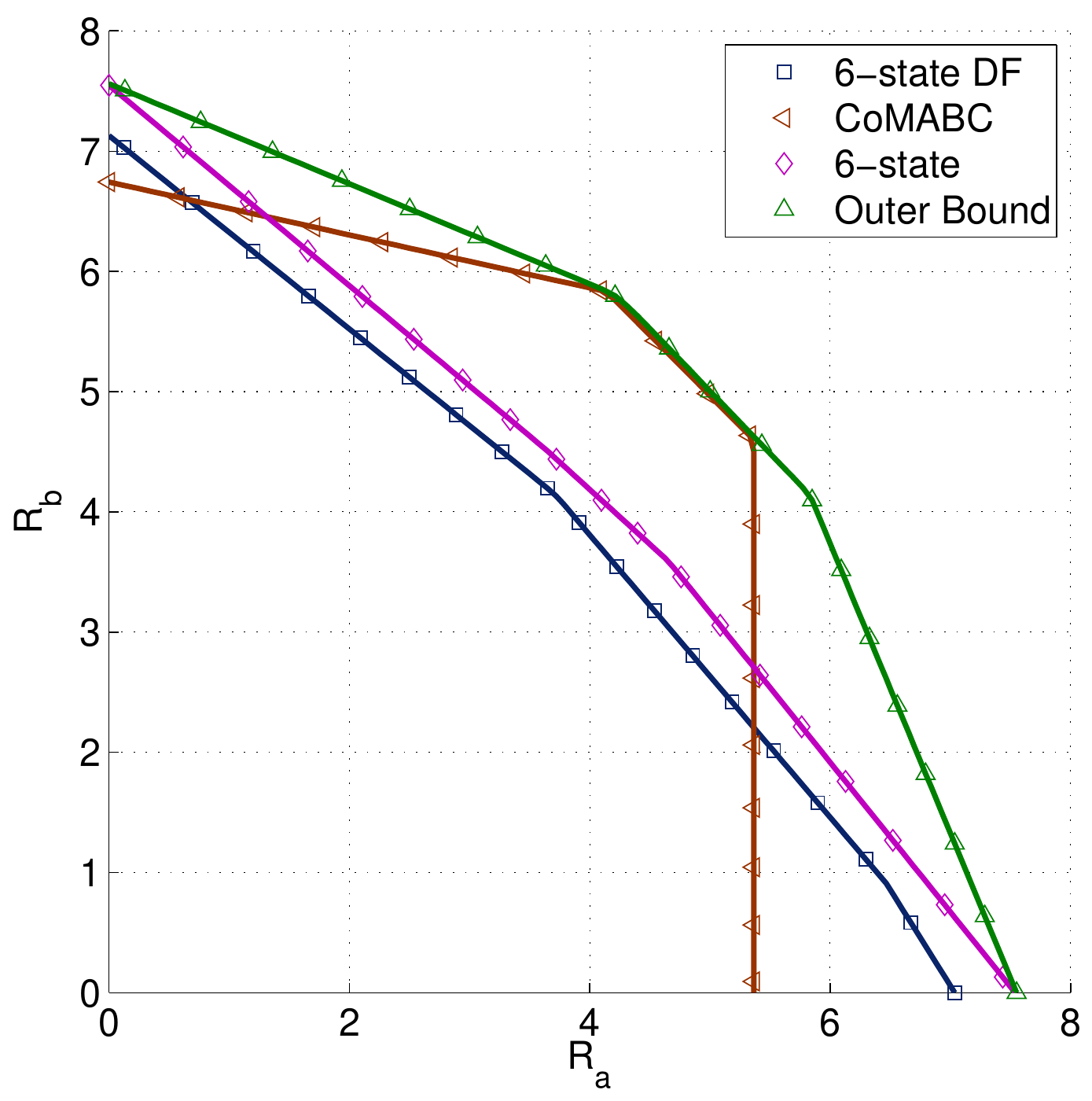}
\caption{Comparison with CoMABC: $\gamma_1 = 30$ dB, $\gamma_2 = 35$ dB, $\gamma_3 = 13$ dB}
\label{figure5}
\end{figure}
\begin{figure}[htb]
\centering
\includegraphics[width=2.4in]{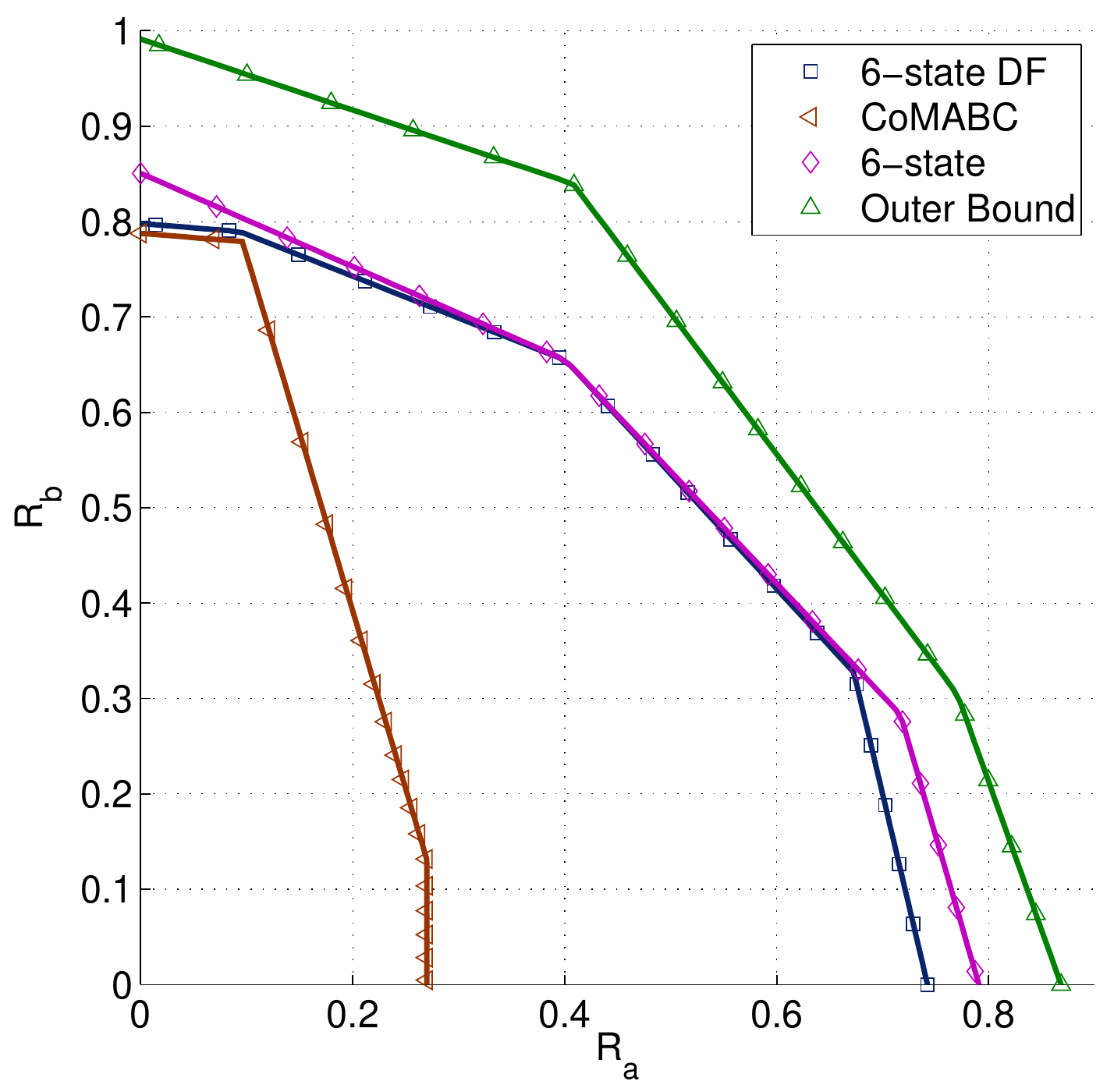}
\caption{Comparison with CoMABC: $\gamma_1 = 0$ dB, $\gamma_2 = 5$ dB, $\gamma_3 = -7$ dB}
\label{figure4}
\end{figure}

\vspace*{-1mm}
\section{Analytical Outer Bounds}
The outer bound in Theorem \ref{th1} can be evaluated by solving a linear program for each $k$. In this section, we derive an analytical expression that is an outer bound by analyzing the dual of the linear programs in Theorems \ref{th1} and \ref{th3}. These outer bounds are easier to evaluate and can also be used to gain a better understanding of the bound. 

\vspace*{-3mm}
\subsection{Outer Bound}
The dual of the linear program in Theorem \ref{th1} is: $\displaystyle{\min_{\{y_i\}} y_5}$
\begin{equation}
{\allowdisplaybreaks
 \begin{split}
\mbox{subject to~}
  y_5 &\geq y_1 \mathcal{C}(\gamma_1+\gamma_3)+ y_2 \mathcal{C}\left(\gamma_3\right),\\
  y_5 &\geq y_3 \mathcal{C}(\gamma_2+\gamma_3)+ y_4 \mathcal{C}\left(\gamma_3\right),\\
  y_5 &\geq y_1 \mathcal{C}(\gamma_1)+ y_3 \mathcal{C}\left(\gamma_2\right),\\
  y_5 &\geq y_2 \mathcal{C}(\gamma_2)+ y_4 \mathcal{C}\left(\gamma_1\right),\\
  y_5 &\geq y_1 \mathcal{C}(\gamma_3)+ y_2 \mathcal{C}\left((\sqrt{\gamma_2} + \sqrt{\gamma_3})^2\right),\\
  y_5 &\geq y_3 \mathcal{C}(\gamma_3)+ y_4 \mathcal{C}\left((\sqrt{\gamma_1} + \sqrt{\gamma_3})^2\right),\\
  &ky_1 + ky_2 + y_3 + y_4 \geq 1.\\
 \end{split}}
\label{constraintsd}
\end{equation}

\begin{theorem}
For any $k \geq 1$, $R_b$ is upperbounded as 
\[
R_b \le \max\{T_1, T_2, T_3, T_4\},
\]
\[
\mbox{where~~}
T_1 = \frac{3k-1}{2k^2}\frac{\mathcal{C}(\gamma_1)\mathcal{C}(\gamma_2)}{\mathcal{C}(\gamma_1) + \mathcal{C}(\gamma_2)},
\]
\[
T_2 = \frac{2k-1}{2k^2}\frac{\mathcal{C}(\gamma_2)\mathcal{C}(\gamma_1 + \gamma_3) + \mathcal{C}(\gamma_1)\mathcal{C}(\gamma_3)}{\mathcal{C}(\gamma_1) + \mathcal{C}(\gamma_2)},
\]
\[
T_3 = \frac{2k-1}{2k^2}\frac{\mathcal{C}(\gamma_2)\mathcal{C}(\gamma_3) + \mathcal{C}(\gamma_1)\mathcal{C}((\sqrt{\gamma_2} + \sqrt{\gamma_3})^2)}{\mathcal{C}(\gamma_1) + \mathcal{C}(\gamma_2)},
\]
\[
T_4 = \frac{1}{2k}\frac{\mathcal{C}(\gamma_1)\mathcal{C}(\gamma_3) + \mathcal{C}(\gamma_2)\mathcal{C}((\sqrt{\gamma_1} + \sqrt{\gamma_3})^2)}{\mathcal{C}(\gamma_1) + \mathcal{C}(\gamma_2)}.
\]
\label{thdual}
\end{theorem}
\begin{proof}
The value of the dual program at any feasible point is an upper bound on the value of the primal problem. Choosing an appropriate feasible point in the dual program provides a good upper bound on $R_b$. The bound on $R_b$ is obtained by choosing the following:
\[
y_1 = \frac{2k-1}{2k^2}\frac{\mathcal{C}(\gamma_2)}{\mathcal{C}(\gamma_1) + \mathcal{C}(\gamma_2)}; \;\;\;\;\; y_2 = \frac{2k-1}{2k^2}\frac{\mathcal{C}(\gamma_1)}{\mathcal{C}(\gamma_1) + \mathcal{C}(\gamma_2)};
\]
\[
y_3 = \frac{1}{2k}\frac{\mathcal{C}(\gamma_1)}{\mathcal{C}(\gamma_1) + \mathcal{C}(\gamma_2)}; \;\;\;\;\; y_4 = \frac{1}{2k}\frac{\mathcal{C}(\gamma_2)}{\mathcal{C}(\gamma_1) + \mathcal{C}(\gamma_2)}
\]
\end{proof}

\begin{corollary}(1) For $k = 1$, $T_2 \le T_4$. Therefore, $R_b \le \max\{T_1, T_3, T_4\}$. (2) For $\gamma_1 = \gamma_2 = \gamma$ and $k = 1$, we get 
\begin{equation}
R_b \le \max\left\{ \frac{\mathcal{C}(\gamma)}{2}, \frac{1}{4}\left[\mathcal{C}(\gamma_3) + \mathcal{C}((\sqrt{\gamma} + \sqrt{\gamma_3})^2) \right]\right\}.
\label{bound1k1}
\end{equation}
\label{corol1}
\end{corollary}

\begin{remark} A result similar to Theorem \ref{thdual} can be obtained for $k < 1$ as well. In this case, we can set $R_b = k'R_a$, where $k' > 1$ and use the same technique as in Theorem \ref{thdual}. The expressions obtained for $T_1$ to $T_4$ are similar to Theorem \ref{thdual} except that $\gamma_1$ and $\gamma_2$ are interchanged in each expression.
\end{remark}

\begin{remark}
$R_b$ is also upper bounded by the upper bound for one-way relaying from $b$ to $a$ with $R_a = 0$, i.e., we have
\[
R_b \le \frac{\mathcal{C}(\gamma_1 + \gamma_3)\mathcal{C}((\sqrt{\gamma_2} + \sqrt{\gamma_3})^2) - \mathcal{C}^2(\gamma_3)}{\mathcal{C}(\gamma_1 + \gamma_3) + \mathcal{C}((\sqrt{\gamma_2} + \sqrt{\gamma_3})^2) - 2\mathcal{C}(\gamma_3)}.\]
This bound is obtained by solving the dual program for one-way relaying. When there is no direct link, i.e., $\gamma_3 = 0$, this reduces to 
$R_b \le \frac{\mathcal{C}(\gamma_1)\mathcal{C}(\gamma_2)}{\mathcal{C}(\gamma_1)+\mathcal{C}(\gamma_2)},$
where the bound is the capacity of the half-duplex two-hop linear network \cite{khojastepour2003capacity}.
\end{remark}

\subsection{Capacity Results}
\begin{theorem}
(1) For $\gamma_1 = \gamma_2 = \gamma$ and $k = 1$, the upper bound on $R_a$ is $\mathcal{C}(\gamma)/2$ for $\gamma_3 \le \gamma_{30}$, where $\gamma_{30}$ satisfies
$ f(\gamma_{30}) = 2 \mathcal{C}(\gamma)$, $f(\gamma_3) \stackrel{\triangle}{=} \mathcal{C}(\gamma_3) + \mathcal{C}((\sqrt{\gamma} + \sqrt{\gamma_3})^2).$\\
(2) For $\gamma_1 \ne \gamma_2$ and $k = 1$, the upper bound on $R_a$ is $\mathcal{C}(\gamma_1)\mathcal{C}(\gamma_2)/(\mathcal{C}(\gamma_1) + \mathcal{C}(\gamma_2))$ for $\gamma_3 \le \min(\gamma_{31}, \gamma_{32})$, where $\gamma_{31}$, $\gamma_{32}$ satisfy 
$ f_1(\gamma_{31}) = f_2(\gamma_{32}) =  2\mathcal{C}(\gamma_1)\mathcal{C}(\gamma_2)$ where $f_1(\gamma_3) \stackrel{\triangle}{=} \mathcal{C}(\gamma_2)\mathcal{C}(\gamma_3) + \mathcal{C}(\gamma_1)\mathcal{C}((\sqrt{\gamma_2} + \sqrt{\gamma_3})^2)$ and $f_2(\gamma_3) \stackrel{\triangle}{=} \mathcal{C}(\gamma_1)\mathcal{C}(\gamma_3) + \mathcal{C}(\gamma_2)\mathcal{C}((\sqrt{\gamma_1} + \sqrt{\gamma_3})^2)$.
\label{captheorem}
\end{theorem}
\begin{proof}
Note that $f(\cdot)$, $f_1(\cdot)$, and $f_2(\cdot)$ are monotonically increasing functions of $\gamma_3$ for a given $\gamma_1$, $\gamma_2$. Therefore: (a) unique solutions exist for $\gamma_{30}$, $\gamma_{31}$, and $\gamma_{32}$, and (b) the resulting $\max$ term in each upper bound in Corollary \ref{corol1} is the first term. 
\end{proof}

In the $\gamma_1 = \gamma_2 = \gamma$, $k=1$ case, the above result means that the upper bound is $\mathcal{C}(\gamma)/2$. The lattice coding scheme in \cite{NamChuLee10} can achieve rates within 0.5 bits of this upper bound without any direct link. Therefore, for the weak direct link regime specified by $\gamma_{30}$, $\mathcal{C}(\gamma)/2$ is the capacity within 0.5 bits and can be achieved without using the direct link. In the $\gamma_1 \ne \gamma_2$, $k=1$ case, the bound reduces to $\mathcal{C}(\gamma_1)\mathcal{C}(\gamma_2)/(\mathcal{C}(\gamma_1) + \mathcal{C}(\gamma_2))$ when the direct link is weak. This is equal to the one-way relaying bound when $\gamma_1 \ne \gamma_2$ without the direct link. 

\subsection{Analytical Upper Bound for Weighted Sum Rate}
\begin{theorem}
For any $k$ , $kR_a + R_b$ is upperbounded as 
\[
kR_a + R_b \le \max\{T_1, T_2, T_3, T_4\},
\]
\[
\mbox{where~~}
T_1 = k \frac{\mathcal{C}(\gamma_1 + \gamma_3)\mathcal{C}((\sqrt{\gamma_2} + \sqrt{\gamma_3})^2) - \mathcal{C}^2(\gamma_3)}{\mathcal{C}(\gamma_1 + \gamma_3) + \mathcal{C}((\sqrt{\gamma_2} + \sqrt{\gamma_3})^2) - 2\mathcal{C}(\gamma_3)},
\]
\[
T_2 =\frac{\mathcal{C}(\gamma_2 + \gamma_3)\mathcal{C}((\sqrt{\gamma_1} + \sqrt{\gamma_3})^2) - \mathcal{C}^2(\gamma_3)}{\mathcal{C}(\gamma_2 + \gamma_3) + \mathcal{C}((\sqrt{\gamma_1} + \sqrt{\gamma_3})^2) - 2\mathcal{C}(\gamma_3)},
\]
\[
T_3 = k\frac{\mathcal{C}(\gamma_1)\mathcal{C}(\sqrt{\gamma_2} + \sqrt{\gamma_3})^2) - \mathcal{C}(\gamma_3)}{\mathcal{C}(\gamma_1+\gamma_3) + \mathcal{C}(\sqrt{\gamma_2} + \sqrt{\gamma_3})^2) - 2\mathcal{C}(\gamma_3)}
\]
\[
+ \frac{\mathcal{C}(\gamma_2)\mathcal{C}(\sqrt{\gamma_1} + \sqrt{\gamma_3})^2) - \mathcal{C}(\gamma_3)}{\mathcal{C}(\gamma_2+\gamma_3) + \mathcal{C}(\sqrt{\gamma_1} + \sqrt{\gamma_3})^2) - 2\mathcal{C}(\gamma_3)}, 
\]
\[
T_4 = k\frac{\mathcal{C}(\gamma_2)\mathcal{C}(\gamma_1 + \gamma_3) - \mathcal{C}(\gamma_3)}{\mathcal{C}(\gamma_1+\gamma_3) + \mathcal{C}(\sqrt{\gamma_2} + \sqrt{\gamma_3})^2) - 2\mathcal{C}(\gamma_3)}
\]
\[
\frac{\mathcal{C}(\gamma_1)\mathcal{C}(\gamma_2 + \gamma_3) - \mathcal{C}(\gamma_3)}{\mathcal{C}(\gamma_2+\gamma_3) + \mathcal{C}(\sqrt{\gamma_1} + \sqrt{\gamma_3})^2) - 2\mathcal{C}(\gamma_3)}. 
\]
\end{theorem}
\begin{proof}
The bound is obtained by writing the dual of the linear program and choosing the dual variables as:
\[
y_1 = \frac{\mathcal{C}((\sqrt{\gamma_2} + \sqrt{\gamma_3})^2) - \mathcal{C}(\gamma_3)}{\mathcal{C}(\gamma_1 + \gamma_3) + \mathcal{C}((\sqrt{\gamma_2} + \sqrt{\gamma_3})^2) - 2\mathcal{C}(\gamma_3)}, 
\]
\[
y_2 = \frac{\mathcal{C}(\gamma_1 + \gamma_3) - \mathcal{C}(\gamma_3)}{\mathcal{C}(\gamma_1 + \gamma_3) + \mathcal{C}((\sqrt{\gamma_2} + \sqrt{\gamma_3})^2) - 2\mathcal{C}(\gamma_3)}, 
\]
\[
y_3 = \frac{\mathcal{C}((\sqrt{\gamma_1} + \sqrt{\gamma_3})^2) - \mathcal{C}(\gamma_3)}{\mathcal{C}(\gamma_2 + \gamma_3) + \mathcal{C}((\sqrt{\gamma_1} + \sqrt{\gamma_3})^2) - 2\mathcal{C}(\gamma_3)}, 
\]
\[
y_4 = \frac{\mathcal{C}(\gamma_2 + \gamma_3) - \mathcal{C}(\gamma_3)}{\mathcal{C}(\gamma_2 + \gamma_3) + \mathcal{C}((\sqrt{\gamma_1} + \sqrt{\gamma_3})^2) - 2\mathcal{C}(\gamma_3)}. \]
\end{proof}
\vspace*{-3mm}
\subsection{Numerical Results}
Figure \ref{figure6} compares the numerical (solving a linear program) and analytical outer bounds for three channel conditions A, B, and C. The two analytical bounds are close to the numerical bound. The first analytical bound is better for symmetric rates, while the second outer bound is better for asymmetric rates. In case B, where $\gamma_1 = \gamma_2$, the first analytical bound matches with numerical bound for $k = 1$ as expected from Theorem \ref{captheorem}.
\begin{figure}[!t]
\centering
\includegraphics[width=3.2in]{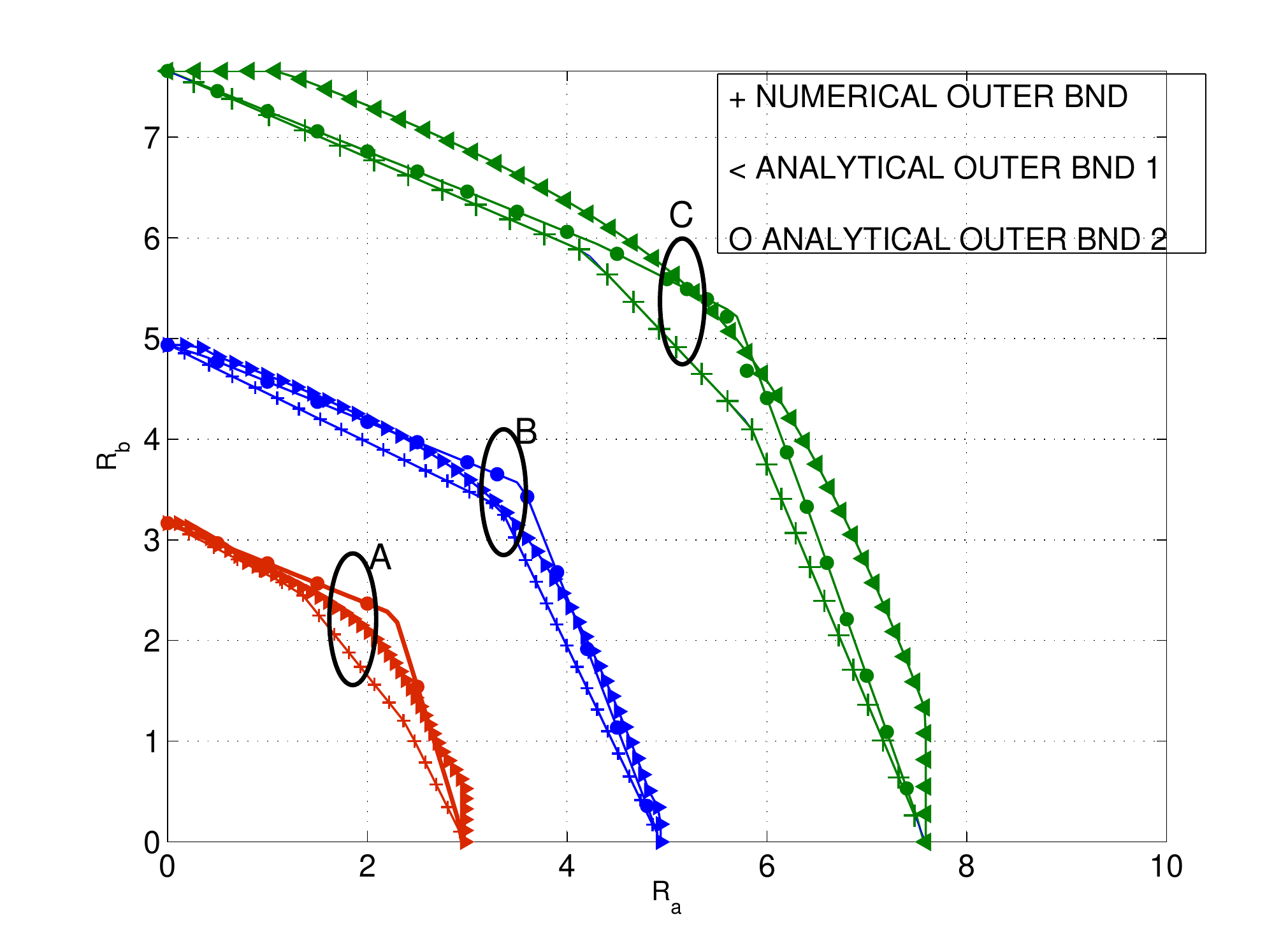}
\caption{Comparison of analytical and numerical outer bounds -- Case A: $\gamma_1 = 10$ dB, $\gamma_2 = 15$ dB, $\gamma_3 = 3$ dB, Case B: $\gamma_1 = 20$ dB, $\gamma_2 = 20$ dB, $\gamma_3 = 8$ dB, Case C: $\gamma_1 = 30$ dB, $\gamma_2 = 35$ dB, $\gamma_3 = 13$ dB,}
\label{figure6}
\end{figure}
Figure \ref{figure7} shows the threshold $\min(\gamma_{31}$, $\gamma_{32})$ defined in Theorem \ref{captheorem} as a function of $\gamma_2$ for a given $c = \gamma_1/\gamma_2$. Three values of $c$ are considered. The results illustrate for these cases the range of $\gamma_3 \le \min(\gamma_{31}$, $\gamma_{32})$ when the direct link can be ignored in the outer bound.  
\begin{figure}[!t]
\centering
\includegraphics[width=3.2in]{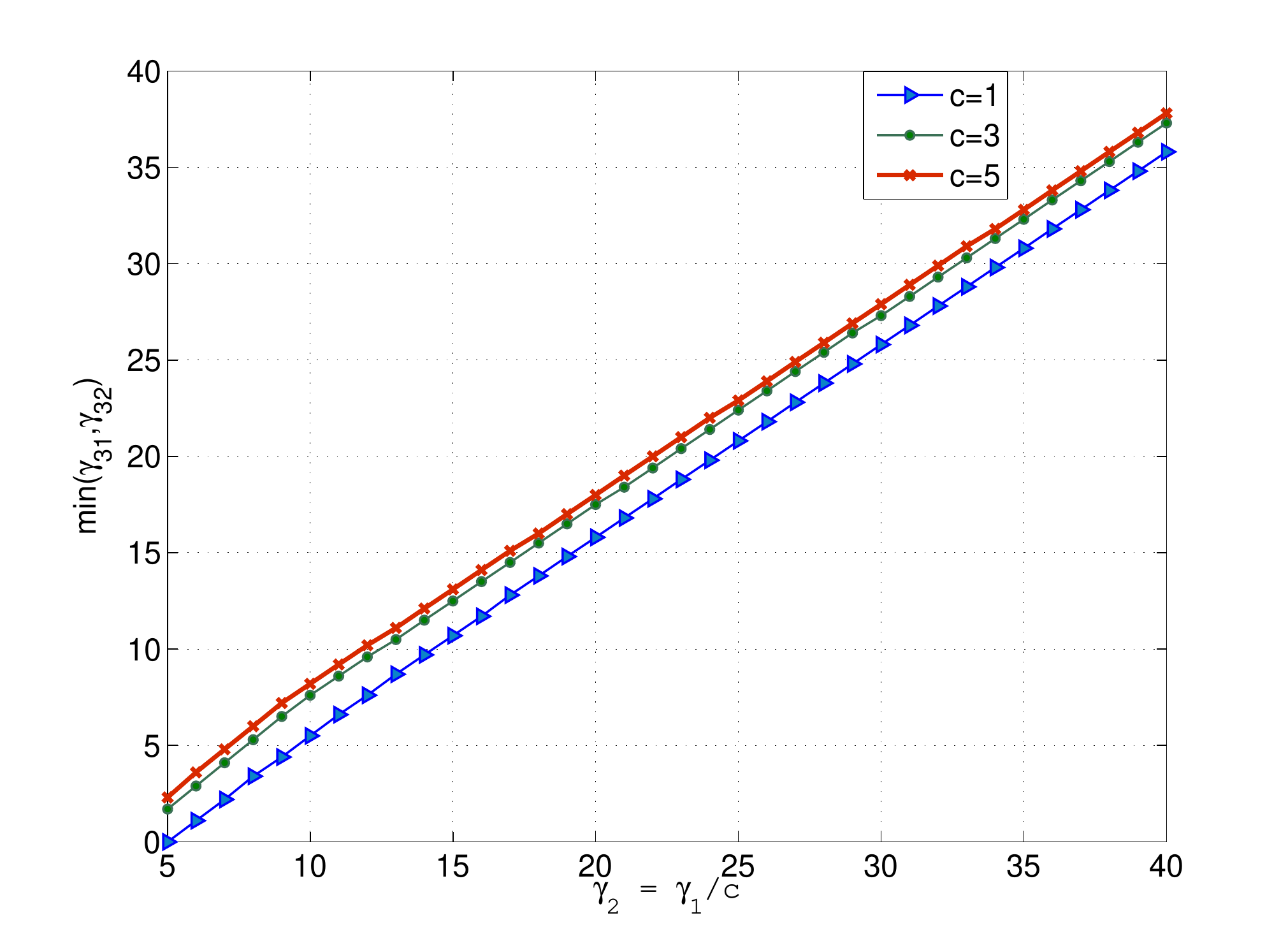}
\caption{Threshold for $\gamma_3 = \min(\gamma_{31},\gamma_{32})$ with $\gamma_1 = c\gamma_2$}
\label{figure7}
\end{figure}

\vspace*{-4mm}
\section{Summary}
The half-duplex two-way Gaussian relay channel with direct link has
been studied. First, an outer bound for the rate region achieved by
any protocol was derived. At any point on this outer bound, only four
out of the six possible states are required. Then, this outerbound was
compared with various existing protocols -- HBC, 6-state, and CoMABC
-- and their achievable rate regions. While the CoMABC protocol is
good near the maximum sum rate region, the 6-state protocol is good
for asymmetric rates, particularly at lower SNRs. The 6-state protocol
is able to achieve this even while being restricted to only DF. The
analytical outer bounds derived in this paper were shown to be close to
the numerically computed outer bound. These analytical outer bounds can
be used to gain a better understanding of the capacity region
boundary. We also obtained the symmetric capacity (to within 0.5 bits)
using the analytical bounds for some channel conditions where the
direct link between nodes a and b is weak. The analytical bounds using
the dual of the linear program could be potentially improved in the
future to identify more channel conditions where capacity can be
determined.

\bibliographystyle{IEEEtran}
\bibliography{IEEEabrv,master}

\end{document}